\documentclass[10pt,a4paper]{llncs}

\usepackage{amssymb,amsfonts,amsmath}
\usepackage {color}
\usepackage{lineno}
\usepackage{verbatim}
\usepackage{fancybox}
\usepackage{soul}
\usepackage{graphicx}
\usepackage{times}
\usepackage{multicol}
\usepackage{ntheorem}

\newtheorem*{proposition*}{Proposition.}

\title{Which convex polyhedra can be made by gluing regular hexagons?
}

\author{Elena Arseneva\inst{1} 
\and Stefan Langerman\inst{2} 
}
\institute{
St. Petersburg State University, 
Saint-Petersburg, Russia
\email{e.arseneva@spbu.ru}
\and
Computer Science Department, Universit\'e libre de Bruxelles (ULB), Belgium
\email{stefan.langerman@ulb.ac.be}
}

\authorrunning{E. Arseneva and S. Langerman}

\begin{document}
\maketitle

\begin{abstract}
Which convex 3D polyhedra can be obtained by gluing several regular hexagons edge-to-edge?
It turns out that there are only 15 possible types of shapes, 
5 of which are doubly-covered 2D polygons.
We give 
examples for most of 
them, including all simplicial and all flat shapes, and give a characterization for the latter ones. 
It is open whether the remaining  can be realized.  
\end{abstract}

\section{Introduction}
Given a 2D polygon $P$, which convex 3D polyhedra can be obtained by folding it and gluing its boundary to itself?
Alexandrov's theorem~\cite{alex} states that for any gluing pattern homeomorphic to a sphere that does not yield 
a total facial angle of more than 
$2\pi$ at any point, there is a unique 3D convex polyhedron that can be constructed in this manner.  
Nevertheless, answering the above question requires checking exponentially many gluing patterns~\cite{DDLO02}. 
Finding the unique 3D polyhedron for a given gluing is a notoriously difficult problem. The best known approximation 
algorithm has pseudopolynomial time complexity~\cite{kpd09-approx}.

There are two ways to restrict the setting: to consider a particular polygon (e.g., all regular polygons~\cite{DO07}, and the Latin cross~\cite{ddlop99}
 were studied), or to consider only  
\emph{edge-to-edge} gluings, 
where an edge of $P$ needs to be glued to an entire other edge of $P$~\cite{DO07,lo96-dynprog}. 

We are interested in gluing together     
several copies of a same regular polygon edge-to-edge, 
thus fusing these two settings, while at the same time extending and restricting each of them. 
The case of regular $k$-gons for $k>6$ is trivial. Indeed, since  gluing three $k$-gons
in one point would violate the above Alexandrov's condition,  
 the only two possibilities are: two $k$-gons glued together and forming 
a doubly covered $k$-gon, or one $k$-gon folded in half (if $k$ is even). 
Thus the first interesting case is $k=6$, and we study it here.  
Note that the
problem we are solving here for $k=6$ is actually decidable (in constant
time) for any constant $k$ by Tarski's theorem,
 but the problem is
probably too large even for $k=6$ to be handled by any existing
computer.

\section{Gaussian Curvature}

Let $P$ be a convex 3D polyhedron. The
Gaussian curvature at a vertex $v$ of $P$ equals $2\pi - \sum_{j=1}^t{\alpha^v_j}$, where $t$ is the number of faces of $P$ incident to $v$, and  $\alpha^v_j$ 
is the angle of the $j$-th face incident to $v$.  Since $P$ is convex, the
Gaussian curvature at each vertex of $P$ is non-negative.

\begin{theorem}[Gauss-Bonnet, 1848]
\label{th:gauss-bonnet}
The total sum of the Gaussian curvature at each vertex of a 3D polyhedron $P$ equals $4\pi$.
\end{theorem}

Let $P$ be a convex polyhedron
 obtained by gluing 
several regular hexagons edge-to-edge. 
 Vertices of $P$ are vertices of the hexagons, and the sum of facial angles around a vertex $v$ of $P$ 
equals $2\pi/3$ (the interior angle of the regular hexagon) times 
the number of hexagons glued together at $v$. Since the Gaussian curvature at $v$ is in $(0,2\pi)$, the number of hexagons glued at $v$ can be either one or two, implying  
the Gaussian curvature of $v$ to be respectively $4\pi/3$ or $2\pi/3$. 
 If three hexagons are glued 
at a point $p$, $p$ has zero Gaussian curvature, and thus is a (flat) point on the surface of $P$. Thus $P$ has at most 6 vertices.

\section{Doubly-covered polygons}
There are $4$ 
combinatorially different  doubly-covered 
plane polygons that can be obtained by gluing hexagons. The quadrilaterals come in 2 variants
depending on the sequence of their angles. Thus we list $5$ types of polygons. 
We list all the types below, and give an example for each type in Figure~\ref{fig:polygons}.

\begin{figure*} 
\centering 
\includegraphics[scale=0.85]{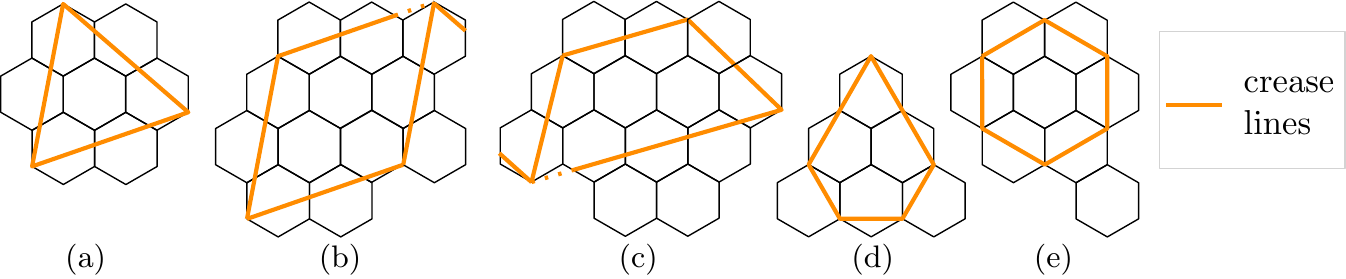}
\caption{Examples of doubly-covered polygons of types (a)-(e):  
Their nets 
and crease lines} 
\label{fig:polygons}
\end{figure*}

\begin{itemize}

\item[(a)] Equilateral triangle. 

\item[(b)]  
Quadrilateral with angles $\pi/3, 2\pi/3, \pi/3, 2\pi/3$. 
 This is an isosceles parallelogram.

\item[(c)] Quadrilateral with angles $\pi/3, \pi/3, 2\pi/3, 2\pi/3$. 
This is a trapezoid.

\item[(d)] Pentagon with 1 angle $\pi/3$, and 4 angles $2\pi/3$. 

\item[(e)] Hexagon with 6 angles $2\pi/3$. 
\end{itemize}

We now give a complete characterization of such shapes:
\begin{theorem}
Polygons of type (a)--(e), that can be drawn on the hexagonal grid, are exactly the polygons, doubly-covered versions of which 
 can be obtained by gluing regular hexagons. 
 
\end{theorem} 
\begin{proof}
Slightly abusing notation, we do not distinguish between polygons and their drawings on the grid.

Consider a polygon $P$ whose doubly-covered version $Q$ can be obtained by gluing regular hexagons.
We can draw on $P$ the hexagons from which it is glued, 
and this will correspond to a drawing of $P$ on the hexagonal grid, i.e, the vertices of $P$ coincide with the vertices of the grid.

Every vertex $v$ of $P$ corresponds to a vertex $v'$ of $Q$, and the internal angle of $P$ at $v$ is exactly half of the total angle of $Q$ at $v'$. Since $Q$ is obtained by gluing regular hexagons edge to edge, the total angle at $v'$ is $2\pi/3$ or $4\pi/3$. Thus every internal angle of $P$ is $\pi/3$ or $2\pi/3$. 
Now let $i$ be the number of vertices in $P$, and let $t$ be the number of vertices of $P$ with internal angle $\pi/3$. 
Since all internal angles of $P$ sum up to $(i-2)\pi$, we obtain an equation
 $\pi/3*t + 2\pi/3(i-t) = (i-2)*\pi$, which for each fixed $i = 3,5,6$ gives us the unique number of angles of value 
$\pi/3$ the polygon has. Because of the symmetry, for $i = 3,5,6$ there is only one possible shape for each of these cases: respectively type (a), (d), and (e).
 The case of $i=4$ has two possibilities, depending on whether the two angles of the same value are adjacent to teach other or not.  For $i>6$ the value of $t$ is negative.  

Now let us prove the other direction of the statement. Consider a polygon $P$ of type (a)-(e), that can be drawn on a hexagonal grid. 
Let $P'$ be a copy of $P$ mirrored with respect to some side $s$ of $P$.  See Figure~\ref{fig:drawing}.
 It is enough to prove that  each vertex of $P'$ coincides with some vertex of the grid, 
and (more strongly) that each side of $P$ and its counterpart in $P'$ break the grid cells exactly the same way (same as above). 

Consider the pairs of corresponding sides of $P$ and $P'$ one by one in the counterclockwise order. 
First, for the side $s$ of $P$ and the side of $P'$ that coincides with $s$, the statement is true by construction. 
The next pair of sides are two line segments rotated w.r.t.\ each other by the angle $2\alpha_1$, 
where $\alpha_1$ is the interior angle of $P$ adjacent to the side $s$ (counterclockwise). 
Further, each $i$-th pair of sides will be rotated by additional value of $2\alpha_i$.  
Since each angle $\alpha_i$ is either $\pi/3$ or $2\pi/3$, each pair of sides is rotated w.r.t. 
each other by the angle $k \cdot 2\pi/3$, for some $k \in \{0,1,2\}$. 
Since the angle of the regular hexagon in $2\pi/3$, the statement holds for each pair of sides.

\qed 
\end{proof}

\begin{figure}[h]
\centering
\includegraphics[scale = 0.75]{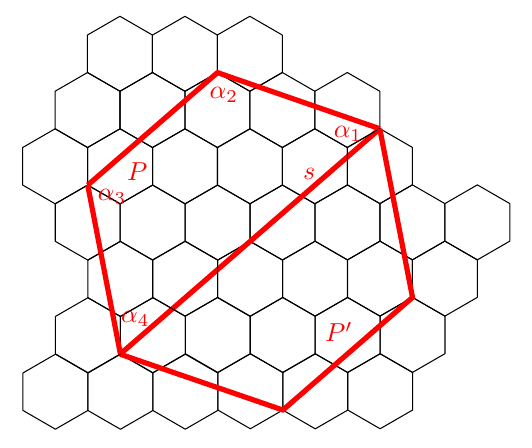}
\caption{Illustration for the proof of the proposition: here $P$ is of type~(c)}
\label{fig:drawing}
\end{figure}

It is interesting to count the number of distinct polygons of a fixed type as a function the number $n$ of hexagons glued to produce the shape. 
Observe that the number in question is polynomial in $n$, because it is composed of two polygons  with at most 6 vertices of diameter at most $n$, drawn on a hexagonal grid. Obtaining a tighter bound is an open problem (see Open Problem~3 in the last section).

\section{Skeletons of Non-flat Polyhedra} 
\begin{figure}[h!]
\centering 
\includegraphics[scale=0.9]{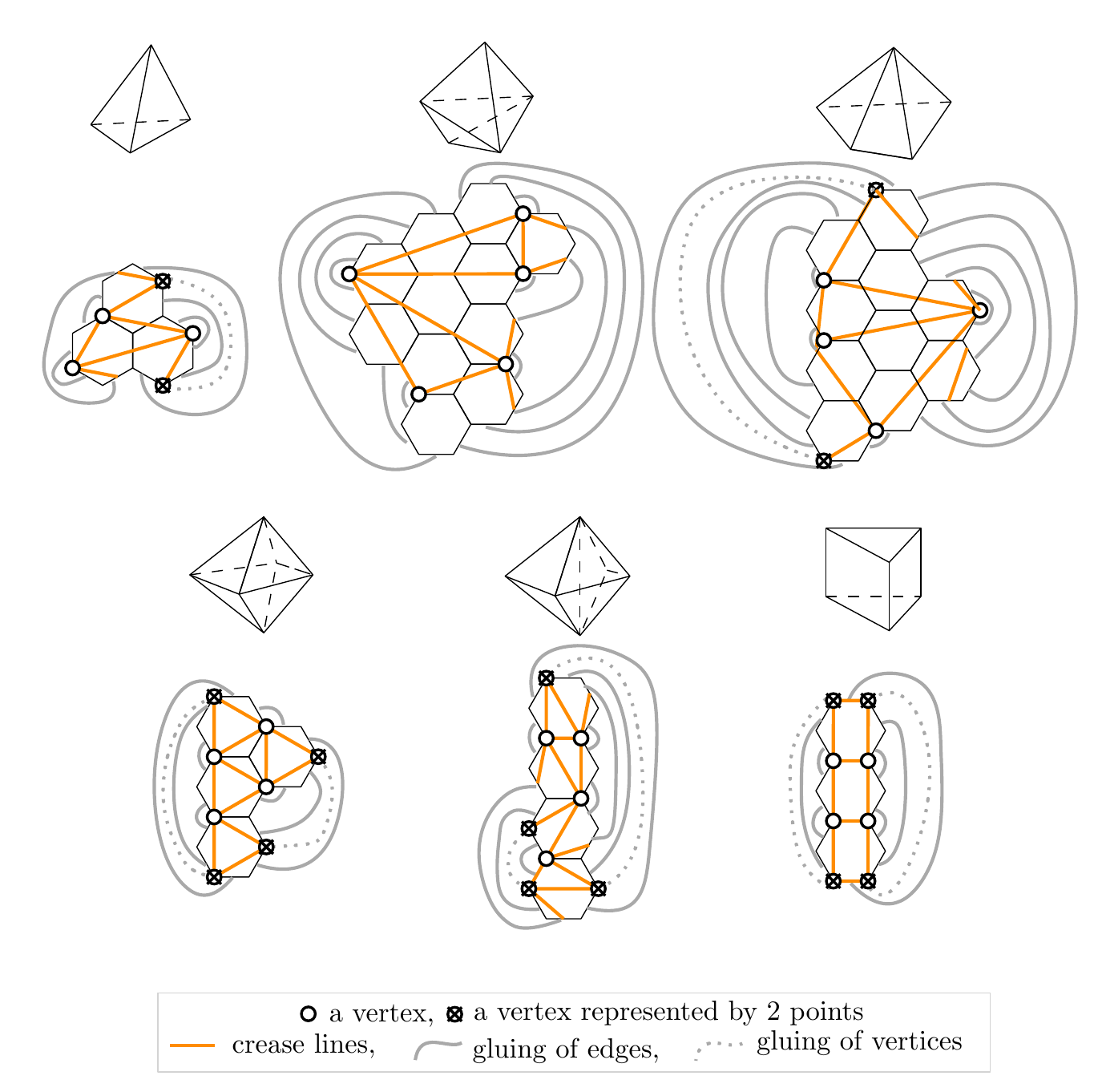}
\caption{Examples of polyhedra (i)--(vi). Above: graphs of their skeletons. 
Below: their nets, crease lines, and gluing rules.}
\label{fig:polyhedra}
\end{figure}

There are $10$ distinct 3-connected simple planar graphs of at most $6$ vertices; these 
are all combinatorially different graph structures of convex polyhedra of at most $6$ vertices.  

Below we give  examples for different polyhedra obtained by gluing regular hexagons.  
Namely we give an example for each doubly-covered flat polygon, and for two 
non-simplicial polyhedra.  
It remains open whether all the non-simplicial polyhedra can be constructed as well (four polyhedra are in question, see Figure~\ref{fig:missing}). Note that we do not characterize these polyhedra in terms of side lengths, as opposed to the case of polygons. Characterization in terms of the Gaussian curvature of the vertices yields another interesting  open question (see Open Question~\ref{open:def-types}).

\begin{itemize}
\item[(i)] Tetrahedron.
\item[(ii)] Hexahedron with 5 vertices (3 vertices of degree 4, and 2 vertices of degree 3), and 6 triangular faces.
\item[(iii)] Pentahedron with 5 vertices (1 vertex of degree 4, 4 vertices of degree 3), one quadrilateral face and 4 triangular faces. In our example, it is a right rectangular pyramid. 
\item[(iv)] Octahedron with 6 vertices of degree 4 each, and 8 triangular faces. In our example, it is a regular octahedron. 
\item[(v)] Octahedron with 6 vertices (two of which are of degree 5, two of degree 4, and two of degree 3) and 8 triangular faces.
\item[(vi)] Pentahedron with 6 vertices of degree 3 each, and 5 faces (2 triangles, and 3 quadrilaterals). In our example, it is a triangular prism. 
\end{itemize}

\section{ Enumerating the gluings of regular hexagons} 
Hoping to get some insight on this problem, we used a computer program to enumerate the non-isomorphic gluings of regular hexagons. 
Observe that the number of gluings for $n$ hexagons is polynomial in $n$. This is because each face is a polygon  with at most 6 vertices drawn on a hexagonal grid of diameter at most $n$
and the number of faces is constant (which is a generalization of the above argument for doubly-covered polygons). It would be interesting to obtain a tighter bound, see Open Problem~3.

We enumerate non-isomorphic gluings  separately for each fixed number $n$ of hexagons. We first produce all the non-isomorphic 
gluings of $n$ hexagons, whose dual graph are trees. Each such gluing is homeomorphic to a disk. 
Then for every gluing,  we iteratively glue together the pairs of consecutive edges on its boundary, for which the vertex separating them has degree three. Such pairs of edges are guaranteed to be glued to each other in every full gluing produced from the given partial gluing. 
After removing the isomorphic gluings produced by the previous step, 
we compute all possible full gluings (i.e., the ones homeomorphic to a sphere). 
For the last step, we use a modification of the dynamic 
programming algorithm to decide whether a given simple polygon has an edge-to-edge gluing~\cite{lo96-dynprog}.

Table~1 summarizes the results of our experiments so far: The second column gives  the number of non-isomorphic gluings of $n$ regular hexagons, for $n$ between 1 and 7. 
The third and forth columns characterize the non-flat shape types for the gluings of at most 4 hexagons.

\begin{table} 
\begin{tabular}{| c |c | l | l|   }
\hline
 $n$ & \hspace{1em} \# of gluings \hspace{1em} & \hspace{1em} \# of non-flat shapes \hspace{1em} & \hspace{1em} types of non-flat shapes \hspace{1em} 
\\
\hline
\hspace{1em} 1 \hspace{1em}   & 2 & 0 & -\\  
 2 & 4 & 1 & (i)\\
 3 & 6 & 3 & (i), (ii), (vi)\\
 4 & 11 & 6 & (i)*2, (ii), (iv)*2, (v) \\
 5 & 10 & 6 & (i)*2, (ii)*2, (v), (vi) \\
 6 & 17 & &\\
 7 & 18 & & \\ \hline
\end{tabular}
\vspace{1em}
\caption{ Results of our experiments to enumerate and analyze gluings of a fixed number of hexagons.}
\end{table}

\paragraph{\bf Open questions.} This paper raises a number of open questions.  
\begin{enumerate}
\item Can the graph structures of convex polyhedra shown in Figure~\ref{fig:missing} be realized by gluing regular hexagons?

\begin{figure}
\centering
\includegraphics[scale=1]{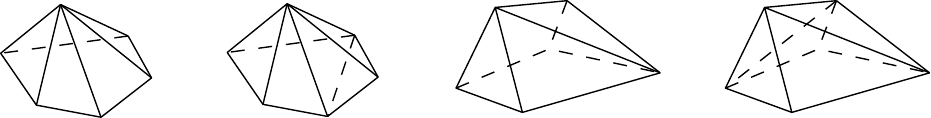}
\caption{The graph structures of convex polyhedra, for which we do not know whether they can be realized by gluing regular hexagons.}
\label{fig:missing}
\end{figure}

\item \label{open:def-types} Are all possible types of vertices according to Gaussian curvature realizable for an graph structure? 
We note that by symmetry of the graph structures and the fact that only two types of vertices exist, this question for the known graph structures reduces to  the question whether the shape of type (iii) is realizable with Gaussian curvature $2\pi/3$ at the vertex that is not incident to the quadrilateral face. See Figure~\ref{fig:def-types}.
 
\begin{figure}
\centering
\includegraphics[scale=0.75]{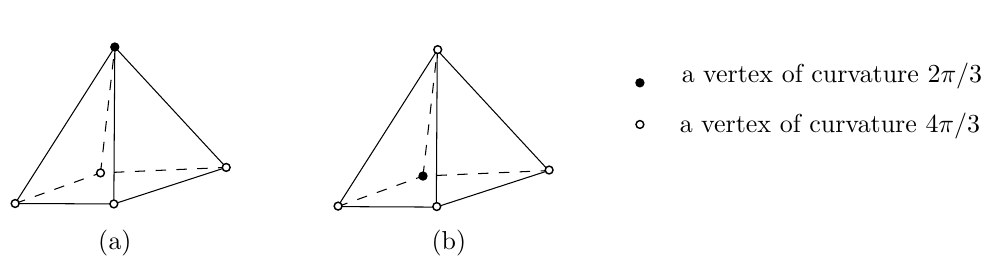}
\caption{Illustration for the open question~\ref{open:def-types}: the shape of type (iii) in the two different variations, (a) our example form Figure~\ref{fig:polyhedra} and (b) the variation for which it is not known whether it is realizable.}
\label{fig:def-types}
\end{figure}

\item 
Bound the number of different 
shapes (flat or non-flat) of a fixed type as a function of the number $n$ of hexagons glued to obtain the shape. 
In the paper we argue that this number is polynomial in $n$, but deriving a tighter bound or even the exact formula is open.  
\end{enumerate}

\paragraph{\bf Acknowledgements.} We wish to thank Jason Ku for providing us the shape (iii) in Figure~\ref{fig:polyhedra},
David Eppstein for posing Open Question~\ref{open:def-types}, and the anonymous referee for valuable comments.

Elena Arseneva was supported in part by F.R.S.-FNRS, 
the SNF Early PostDoc Mobility grant P2TIP2-168563, and she is a research group member of a winner of the competition ``Leader'' 2019 of the Foundation for the Advancement of Theoretical Physics and Mathematics ``BASIS''. 
Stefan Langerman is directeur de recherches du F.R.S.-FNRS.

\end{document}